\documentclass[envcountsame,envcountsect,12pt, a4paper]{scrartcl}

\usepackage{amssymb, amsmath}
\usepackage{graphicx}
\usepackage{xspace}

\usepackage[affil-it]{authblk}

\usepackage{amsthm}
\theoremstyle{plain}
\newtheorem{theorem}{Theorem}[section]
\newtheorem{lemma}[theorem]{Lemma}
\newtheorem{proposition}[theorem]{Proposition}
\newtheorem{corollary}[theorem]{Corollary}

\theoremstyle{definition}

\newtheorem{definition}[theorem]{Definition}

\theoremstyle{remark}

\usepackage{mathdots}

\usepackage{pgf}
\usepackage{pgfplots}
\usepackage{tikz}
\usetikzlibrary{arrows,positioning,automata}

\tikzset{node distance=2.8cm, auto, semithick}
\tikzstyle{every state}=[fill=black!15, node distance=3.1cm]

\usepackage{algorithm}
\usepackage{algorithmicx}
\usepackage[noend]{algpseudocode}

\usepackage{complexity}

\usepackage{mdframed}

\def\phi{\varphi}
\def\NAT{\mathbb{N}}

\def\td{\mathit{td}}

\def\edd{\mathop{ed}_\C}

\def\C{\mathcal{C}}

\def\H{\mathcal{H}}

\def\cover{\mathcal{X}}
\def\str{\mathop{str}}
\def\apex#1{{#1}^{\text{apex}}}

\newcommand{\CC}{\C}

\newcommand{\dist}{\ensuremath{\mathrm{dist}}}

\title{Fixed-parameter Tractable Distances to Sparse Graph Classes}
\author{Jannis Bulian}
\author{Anuj Dawar}
\affil{University of Cambridge Computer Laboratory}

\begin{document}

\maketitle

\begin{abstract}
We show that for various classes $\CC$ of sparse graphs, and several measures of distance to such classes (such as edit distance and elimination distance), the problem of determining the distance of a given graph $G$ to $\CC$ is fixed-parameter tractable.  The results are based on two general techniques.  The first of these, building on recent work of Grohe et al.\ establishes that any class of graphs that is slicewise nowhere dense and slicewise first-order definable is $\FPT$.  The second shows that determining the elimination distance of a graph $G$ to a minor-closed class $\CC$ is $\FPT$.

\end{abstract}

\section{Introduction}

The study of parameterized algorithmics for graph problems has thrown
up a large variety of structural parameters of graphs.  Among these
are parameters that measure the \emph{distance} of a graph $G$ to a
class $\CC$ in some way.  The simplest such measures are those that
count the number of vertices or edges that one must delete (or add) to
$G$ to obtain a graph in $\CC$.  A common motivation for studying such
parameters is that if a problem one wishes to solve is tractable on
the class $\CC$, then the distance to $\CC$ provides an interesting
parameterization of that problem (called \emph{distance to triviality}
by Guo et al.~\cite{guo_structural_2004}).  Other examples of this include the study of modulators to graphs of bounded tree-width in the context of kernelization (see~\cite{FLMS12,Gajarsky13}) or the parameterizations of colouring problems (see~\cite{Marx06}).
On the other hand, determining the distance of an input graph $G$ to a class $\CC$ is, in general, a computationally challenging problem in its own right.  Such problems have also been extensively studied with a view to establishing their complexity when parameterized by the distance.  A canonical example is the problem of determining the size of a minimum vertex cover in a graph $G$, which is just the vertex-deletion distance of $G$ to the class of edge-less graphs.
More generally, Cai~\cite{Cai96} studies the parameterized complexity of distance measures defined in terms of addition and deletion of vertices and edges to hereditary classes $\CC$.  Counting deletions of vertices and edges gives a rather simple notion of distance, and many more involved notions have also been studied.  Classic examples include the crossing number of a graph which provides one notion of distance to the class of planar graphs or the treewidth of a graph which can be seen as a measure of distance to the class of trees.  Another recently introduced measure is \emph{elimination distance}, defined in~\cite{BD14} where it was shown that graph isomorphism is $\FPT$ when parameterized by elimination distance to a class of graphs of bounded degree.

In this paper we consider the fixed-parameter tractability of a
variety of different notions of distance to various different classes
$\CC$ of sparse graphs.  We establish two quite general techniques for
establishing that such a distance measure is $\FPT$.  The first builds
on the recent result of Grohe et al.~\cite{Grohe:2014ec} which shows
that the problem of evaluating first-order formulas on
any \emph{nowhere dense} class of graphs is $\FPT$ with the formula as
parameter.  We extract from their proof of this result a general
statement about the fixed-parameter tractability of definable sparse
classes.  To be precise, we show that parameterized problems that are
both \emph{slicewise nowhere dense} and \emph{slicewise first-order
definable} (these terms are defined precisely below) are $\FPT$.  As
an application of this, it follows that if $\CC$ is a nowhere dense
class of graphs that is definable by a first-order formula, then the
parameterized problem of determining the distance of a graph $G$ to
$\CC$ is $\FPT$, for various notions of distance that can be
themselves so defined.  In particular, we get that various forms of
edit distance to classes of bounded-degree graphs are $\FPT$ (a result
established by Golovach~\cite{Golovach:2014eq} by more direct
methods).  Another interesting application is obtained by considering
elimination distance of a graph $G$ to the class $\CC$ of empty
graphs.  This is nothing other than the \emph{tree-depth} of $G$.
While elimination distance to a class $\CC$ is in general not
first-order definable, it is in the particular case where $\CC$ is the
class of empty graphs.  Thus, we obtain as an application of our
method the result that tree-depth is $\FPT$, a result previously known
from other algorithmic meta theorems (see~\cite[Theorem~17.2]{sparsity}).  The method of establishing that a parameterized problem is $\FPT$ by establishing that it is slicewise nowhere dense and slicewise first-order definable appears to be a powerful method of some generality which will find application beyond these examples.

Our second general method specifically concerns elimination distance to a minor-closed class $\CC$.  We show that this measure is fixed-parameter tractable for any such $\CC$, answering an open question posed in~\cite{BD14}.  Note that while a proper minor-closed class is always nowhere dense, it is not generally first-order definable (for instance, neither the class of acyclic graphs nor the class of planar graphs is), and elimination distance to such a class is also not known to be first-order definable.  Thus, our results on the tractability of slicewise first-order definable classes do not apply here.  Instead, we build on work of Adler et al.~\cite{Adler:2008wm} to show that from a finite list of the forbidden minors characterising $\CC$, we can compute the set of forbidden minors characterising the graphs at elimination distance $k$ to $\CC$.  Adler et al.\ show how to do this for apex graphs, from which one immediately obtains the result for graphs that are $k$ deletions away from $\CC$.  To extend this to elimination distance $k$, we show how we can construct the forbidden minors for the closure of a minor closed class under disjoint unions.

In Section~\ref{S:prelim} we present the definitions necessary for the
rest of the paper.  Section~\ref{S:nd} establishes our result for
slicewise first-order definable and slicewise nowhere dense problems
and gives some applications of the general method.
Section~\ref{S:minor} establishes that the problem of determining
elimination distance to any minor-closed class is $\FPT$.  Some open
questions are discussed in Section~\ref{S:conclusion}.

\section{Preliminaries}
\label{S:prelim}

\paragraph{First-order logic.}
We assume some familiarity with first-order logic for Section~\ref{S:nd}.
A \emph{(relational) signature} $\sigma$ is a finite set of relation symbols, each with an associated arity.
A \emph{$\sigma$-structure} $A$ consists of a set $V(A)$ and for each $k$-ary relation symbol $R \in \sigma$
a relation $R(A) \subseteq V(A)^k$. Our structures will mostly be (coloured) graphs, so $\sigma = \{E\}$ or
$\sigma = \{E, C_1, C_2, \dots, C_r\}$ where $E$ is binary and the $C_i$ are unary relation symbols. A graph $G$
is then a $\sigma$-structure with vertex set $V(G)$, edge relation $E(G)$, and colours $C_i(G)$.

A first-order formula $\phi$ is recursively defined by the following rules:
\[
\phi := R(x_1, \dots, x_r) \mid
  x = y  \mid
  \lnot \phi \mid
  \phi \lor \phi \mid
  \exists x . \phi.
\]
We also use the following abbreviations:
\[
\phi \land \psi := \lnot (\lnot \phi \lor \lnot \psi), \qquad \forall x . \phi := \lnot \exists . \lnot \phi.
\]
The \emph{quantifier rank} of a formula $\phi$ is the nesting depth of quantifiers in $\phi$.
For a more detailed presentation we refer to Hodges~\cite{Hodges:1997we}.

\paragraph{Parameterized Complexity.}
Parameterized complexity theory is a two-dimensional approach to the
study of the complexity of computational problems.  We find it
convenient to define problems as classes of structures rather than
strings. 
A \emph{problem} $Q \subseteq \str(\sigma)$ is an (isomorphism-closed)
class of
$\sigma$-structures given some signature $\sigma$. A
\emph{parameterization} is a function $\kappa : \str(\sigma) \to
\mathbb{N}$. We say that $Q$ is \emph{fixed-parameter tractable} with
respect to $\kappa$ if we can decide whether an input $A \in
\str(\sigma)$ is in $Q$ in time $O(f(\kappa(A)) \cdot |x|^c)$, where
$c$ is a constant and $f$ is some computable function.
For a thorough discussion of the subject we refer to the books by
Downey and Fellows~\cite{Downey:2012vk}, Flum and
Grohe~\cite{Flum:2006vj} and Niedermeier~\cite{Niedermeier:2006ei}.

A parameterized problem $(Q, \kappa)$ is \emph{slicewise first-order definable} if
there is a computable function $f$ such that:
\begin{itemize}
\item for all $i \in \NAT$, we have that $f(i) = \phi_i \in $ FO$[\sigma]$;
\item a $\sigma$-structure $A$ with $\kappa(A) \leq i$ is in $Q$ if and only if $A \models \phi_i$.
\end{itemize}
Slicewise definability of problems in a logic was introduced by Flum and Grohe~\cite{Flum:2001kc}.

\paragraph{Graph theory.}
A \emph{graph} $G$ is a set of vertices $V(G)$ and a set of edges $E(G)
\subseteq V(G) \times V(G)$.
We assume that graphs are loop-free and undirected,
i.e.\ that $E$ is irreflexive and symmetric.
We
mostly follow the notation in Diestel~\cite{Diestel:2000vm}.  For a
set $S \subseteq V(G)$ of vertices, we write $G \setminus S$ to denote
the subgraph of $G$ induced by $V(G)\setminus S$.

Let $r \in \NAT$. An $r$-independent set in a graph $G$ is a set of
vertices of $G$ such that their pairwise distance is at least $r$.

A graph $H$ is a \emph{minor} of a graph $G$, written $H \preceq G$, if
there is a map, called the \emph{minor map}, that takes each vertex
$v \in V(H)$ to a tree
$T_v$ that is a subgraph of $G$ such that for any $u \neq v$ the trees are disjoint, i.e. $T_v \cap T_u = \emptyset$,
and such such that for every edge $uv \in E(H)$ there are vertices $u' \in T_u, v' \in T_v$ such that $u'v' \in E(G)$.
A class of graphs $\C$ is \emph{minor-closed} if $H \preceq G \in \C$ implies $H \in \C$.

The \emph{set of minimal excluded minors} $M(\C)$ is the set of graphs in the complement of $\C$ such
that for each $G \in M(\C)$ all proper minors of $G$ are in $\C$.
By the Robertson-Seymour Theorem~\cite{RS-GMXX} the set $M(\C)$ is finite for every minor-closed class $\C$.
It is a consequence of this theorem that membership in a minor-closed
class can be tested in $O(n^3)$ time.

Let $r \in \NAT$. A minor $H$ of $G$ is a \emph{depth-$r$ minor of
$G$} if there is a minor map that takes vertices in $H$ to trees that have radius at  
most $r$.  
 A class of graphs $\C$ is \emph{nowhere dense} if for every $r \in
\NAT$ there is a graph $H_r$ such that for no $G \in \C$ we have $H_r
\preceq_r G$. A nowhere dense class of graphs $\C$ is called
\emph{effectively nowhere dense} if there is a computable function $f$
from integers to graphs such that for no $G \in \C$ and no $r$ we have
$f(r) \preceq_r G$. 
We are only interested in effectively nowhere dense classes so we
simply use the term \emph{nowhere dense} to mean effectively nowhere dense.

We say that a parameterized graph problem $(Q, \kappa)$ is \emph{slicewise nowhere dense} if there
is a computable function $h$ from pairs of integers to graphs such that for all $i \in \NAT$, we have for no
$G \in \{ H \in Q \mid \kappa(H) \leq i \}$ and $r$ that $h(i, r) \preceq_r G$. We will call $h$ the
\emph{parameter function} of $Q$.

For a class of graphs $\C$ we denote the closure of $\C$ under taking disjoint unions by $\overline\C$.
We say that a graph $G$ is an \emph{apex graph} over a class
$\C$ of graphs if there is a vertex $v \in V(G)$ such that the graph $G
\setminus \{v\} \in \C$. The class of
all apex graphs over $\C$ is denoted $\apex{\C}$.

A graph $G$  has \emph{deletion distance $k$ to a class $\C$} if
there are $k$ vertices $v_1, \dots, v_k \in V(G)$ such that $G
\setminus \{v_1, \dots, v_k\} \in \C$.

The \emph{elimination distance} of a graph $G$  to a class $\C$
defined as follows:
{ \small
\[
\textstyle{\edd(G)} :=
\begin{cases}
  0,
  & \text{if }G \in \C; \\
  1 + \min \{\edd(G \setminus v) \mid v \in V(G)\},
  & \text{if $G \not\in \C$ and $G$ is connected;} \\
  \max\{\edd(H) \mid H \text{ a connected component of $G$}\},
  & \text{otherwise.}
\end{cases}
\]
}

\section{A general method for editing distances}
\label{S:nd}

In this section we establish a general technique for showing that certain definable parameterized problems on graphs are $\FPT$.  As an application, we show that certain natural distance measures to sparse graph classes are $\FPT$.   To be precise, we show that if a parameterized problem is both slicewise
first-order definable and slicewise nowhere dense, then it is $\FPT$.
In particular, this implies that if we have a class $\CC$ that is first-order definable and nowhere dense and the distance measure we are interested in is also first-order definable (that is to say, for each $k$ there is a formula that defines the graphs at distance $k$ from $\CC$), then the problem of determining the distance is $\FPT$.  More generally, if we have a parameterized problem $(Q,\kappa)$ that is slicewise nowhere dense and slicewise first-order definable, and a measure of distance to it is definable in the sense that for any values of $k$ and $d$, there is a first-order formula defining the graphs at distance $d$ to the class $\{G \mid G \in Q \mbox{ and } \kappa(G) \leq k\}$, then the problem of deciding whether a graph has distance at most $d$ to this class is $\FPT$ parameterized by $d+k$.  In particular, this yields the result of Golovach~\cite{Golovach:2014eq} as a consequence.

The method is an adaption of the main algorithm in Grohe\emph{ et
  al.}~\cite{Grohe:2014ec}.   Since the proof is essentially a modification of their central construction, rather than give a full account, we state the main results they prove and explain briefly how the proofs can be adapted for our purposes.  For a full proof, this section is  best read in conjunction with the paper~\cite{Grohe:2014ec}.   Section~\ref{subsec:nd} gives an overview of the key elements of the construction from~\cite{Grohe:2014ec} and the elements from it which we need to extract for our result.  Section~\ref{subsec:algorithm} then gives our main result and Section~\ref{subsec:applications} derives some consequences for distance measures. 

\subsection{Evaluating Formulas on Nowhere Dense Classes}\label{subsec:nd}

 The key result of \cite{Grohe:2014ec} is:

\begin{theorem}[Grohe\emph{ et al.}~\cite{Grohe:2014ec}, Theorem 1.1] \label{T:main_nd}
For every nowhere dense class $\C$ and every $\epsilon > 0$, every
property of graphs definable in first-order logic can be decided in
time $O(n^{1+\epsilon})$ on $\C$.
\end{theorem}

We first give a sketch of the algorithm from
Theorem~\ref{T:main_nd} with an emphasis on the changes needed for our
purposes. We refer to \cite{Grohe:2014ec} for several definitions and results.


The algorithm developed in the proof of Theorem~\ref{T:main_nd} uses a locality-based approach, similar to that used by Frick and Grohe~\cite{Frick:1999cd} to show that first-order evaluation is $\FPT$ on graphs of local bounded treewidth and developed in~\cite{DGK07} for application to graph classes with locally excluded minors .  The idea is that any first-order formula $\phi$ is, by Gaifman's theorem, equivalent to a Boolean combination of \emph{local}
formulae, that is formulae that assert the existence of neighbourhoods satisfying certain conditions.  In classes of sparse graphs where the size (or other parameter) of neighbourhoods of a given radius can be bounded, this yields an efficient evaluation algorithm.

In nowhere dense classes of graphs, we cannot in general bound the size of neighbourhoods.   For example, the class of apex graphs is nowhere dense, but a graph may contain a vertex whose neighbourhood is the whole graph.  However, nowhere dense classes are \emph{quasi-wide}~\cite{sparsity}, which means that we can remove a small (i.e.\ parameter-dependent) set of vertices  (the bottleneck) to ensure that there are many vertices that are far away from each other.  Grohe\emph{ et al.}~\cite{Grohe:2014ec} use this
approach to iteratively transform the input graph into a coloured graph where key bottleneck vertices are removed and vertices are coloured to keep relevant information.  At the same time the formula $\phi$ to be evaluated is also transformed so that it can be evaluated on the modified graph.  
This procedure terminates on nowhere dense classes of graphs within a
constant number of steps.

A key data structure used in the algorithm is a neighbourhood
cover, i.e.\ a collection of connected subgraphs, called clusters, so
that each neighbourhood of a vertex is contained in one of the
clusters. The radius of a cover is the maximum radius of any of its
clusters. The degree of a vertex in a cover is the number of clusters
the vertex is contained in. An important result from
\cite{Grohe:2014ec} is that graphs from a nowhere dense class allow
for small covers and that such a cover can be efficiently computed.

\begin{theorem}[Grohe\emph{ et al.}~\cite{Grohe:2014ec}, Theorem~6.2] \label{T:cover} 
Let $\C$ be a nowhere dense class of graphs. There is a function $f$ 
such that for all $r \in \NAT$ and $\epsilon > 0$ and all graphs $G  
\in \C$ with $n \geq f(r, \epsilon)$ vertices, there exists an  
$r$-neighbourhood cover of radius at most $2r$ and maximum degree at  
most $n^\epsilon$ and this cover can be computed in time $f(r,
\epsilon) \cdot n^{1+ \epsilon}$. Furthermore, if $\C$ is effectively  
nowhere dense, then $f$ is computable.  
\end{theorem}

In this theorem, $f$ is a function of $r$ and $\epsilon$ and depends on the class $\C$ in the sense that it is determined, for an effectively nowhere dense $\C$ by its parameter function.  To be precise, the algorithm needs to order 
the vertices of $G$ to witness a weak colouring number of less than
$n^\epsilon$. The weak colouring number is an invariant of the graph
that is guaranteed to be low for graphs from a nowhere dense class.
The time bound $f(r, \epsilon)\cdot n^{1+\epsilon}$ is obtained using an algorithm for this from Nesetril and Ossona de Mendez~\cite{Nesetril:2008ev}.

While the algorithm of~\cite{Grohe:2014ec} assumes that the input graph $G$ comes from the class $\C$, we can say something more.  For a fixed nowhere dense class $\CC$, where we know the parameter function $h$, we can, given $G$, $r$ and $\epsilon$, compute a bound on the running time of the algorithm from Theorem~\ref{T:cover} .  By running the algorithm to this bound, we have the following as a direct consequence of the proof of Theorem~\ref{T:cover} .

\begin{lemma} \label{L:cover}
There is a computable function $f$ and an algorithm $A$ which given
any graph $G$ with $n$ vertices and for any  $r \in \NAT$ and $\epsilon > 0$ either
computes  an  
$r$-neighbourhood cover of radius at most $2r$ and maximum degree at  
most $n^\epsilon$ or determines that $G \not\in \C$.
\end{lemma}
 


At the core of the proof of  Theorem~\ref{T:main_nd} is the Rank-Preserving Locality Theorem.
Given a neighbourhood cover $\cover$ in a graph $G$, the algorithm iteratively removes bottleneck vertices and 
adds colours to the neighbourhoods of removed vertices to obtain a
coloured graph denoted $G \star_{\cover}^{q+1} q$.  Here $q$ is an integer parameter obtained from the first-order formula $\phi$ that we wish to evaluate in $G$.  At the same time, $\phi$ is transformed into a formula $\hat \phi$ that is \emph{(a)} in the expanded signature of the $G \star_{\cover}^{q+1} q$; and \emph{(b)} in a logic $\FO^+$ which enriches $\FO$ by allowing us to assert distances between vertices without the need for quantifiers.  This ensures that the \emph{local} sentence $\hat \phi$ has the same quantifier rank as $\phi$, giving us the Rank-Preserving Locality Theorem below.  In the following statement, a \emph{$(q+1, r)$-independence sentence} is a formula asserting the existence of a distance-$r$-independent set of size $q+1$ of a particular colour.

\begin{theorem}[Rank-Preserving Locality Theorem, Grohe\emph{ et
    al.}~\cite{Grohe:2014ec}, Theorem 7.5] \label{T:rank_local}
For every $q \in \NAT$ there is an $r$ such that for every $\FO$-formula 
$\phi(x)$ of quantifier rank $q$ there is an $\FO^+$-formula $\hat \phi(x)$, which is a Boolean combination  of $(q+1, r)$-independence sentences and atomic formulas, such that for any graph $G$ every $r$-neighbourhood cover $\cover$
of $G$, and every $v \in V(G)$,
\[
G \models \phi(v) \iff G \star^{q+1}_{\cover} q \models \hat \phi(v). 
\]
Furthermore, $\hat \phi$ is computable from $\phi$, and $r$ is computable from $q$.
\end{theorem}

An important tool for constructing $G \star_{\cover}^{q+1} q$ is a game
characterisation of nowhere dense classes. The game has
three parameters: $\ell, m, r$. In the
$(\ell, m, r)$ Splitter game two players Connector and Splitter play against each
other. In each round Connector chooses a vertex $u$, and Splitter has to
respond with a set $A$ of vertices of size at most $m$ in the $r$-neighbourhood of $u$. In the next round the graph is the neighbourhood of $u$
with the vertices from $A$ removed.  If the graph is empty, Splitter wins. If
Connector survives for more than $\ell$ rounds, she wins.
Grohe\emph{ et al.}~\cite{Grohe:2014ec} prove~\cite[Theorem~4.2]{Grohe:2014ec} that if $\C$ is a nowhere dense class, then there
are $\ell, m$ such that Splitter has a winning strategy on the $(\ell,
m, 2r)$ Splitter game on every graph in $\C$.  

The Splitter's strategy on a graph $G$ 
(which can be efficiently computed) is the essential tool in the construction of $G
\star_{\cover}^{q+1} q$.
The inductive procedure used to compute $G
\star_{\cover}^{q+1} q$ from $G$  is outlined in \cite[Proof of Theorem
8.1]{Grohe:2014ec}.   We note that the termination of the 
algorithm depends on the length of the game -- which is bounded by a
constant since $\C$ is nowhere dense.  The strategy to compute
Splitter's moves is described in
\cite[Remark~4.3]{Grohe:2014ec}.  Since the run time of the algorithm
to compute  $G \star_{\cover}^{q+1} q$ only depends on $q$ and the
length of the Splitter game and we can compute this in advance, we can once again extract the fact that if we start with an arbitrary graph $G$, we can efficiently \emph{either} transform it into $G \star_{\cover}^{q+1} q$ \emph{or} determine that it is not in the class $\C$.  This is summed up in the following lemma. 

\begin{lemma} \label{L:structure}
Let $\C$ be a nowhere dense class of graphs. There is an algorithm
that runs in time $O(q)$ which given a graph $G$
returns $G \star_{\cover}^{q+1} q$ or determines that $G \not\in \C$.
\end{lemma}

Theorem~\ref{T:rank_local} reduces the problem of evaluating a formula of first-order logic to deciding a series of distance-$r$-independent set problems.  So, the final ingredient is to show that this is tractable.  Formally, the problem is defined as follows:

\begin{mdframed}
 \textsc{Distance Independent Set}\\
 \textbf{Input: } A graph $G$ and $k, r \in \NAT$.\\
\textbf{Parameter:} $k+r$\\
 \textbf{Problem: } Does $G$ contain an $r$-independent set of size $k$?
\end{mdframed}

The problem is shown to be $\FPT$ on nowhere dense classes of graphs
\cite[Theorem~5.1]{Grohe:2014ec}.  Since the runtime of the algorithm
depends on the length of the Splitter game and Splitter's strategy, and this can be bounded in advance,
\cite[Theorem~5.1]{Grohe:2014ec} can be restated as follows:

\begin{lemma} \label{L:indset}
Let $\C$ be a nowhere dense class of graphs. Then there is an
algorithm $A$ and a computable function $f$ such that for every $\epsilon > 0$
$A$ runs in time $f(\epsilon, r, k)$ and either solves the
\textsc{Distance Independent Set} problem or determines that $G \not\in \C$.
\end{lemma}

This is all we need to evaluate $\hat\phi$ on $G \star_{\cover}^{q+1}
q$, which is equivalent to evaluating $\phi$ on $G$ by Theorem~\ref{T:rank_local}.

\subsection{Deciding Definable nowhere dense Problems}\label{subsec:algorithm}

The main result of \cite{Grohe:2014ec} establishes that checking whether $G\models \phi$ is $\FPT$ when parameterized by $\phi$ provided that $G$ comes from a known nowhere dense class $\C$.  Thus, the formula is arbitrary, but the graphs come from a restricted class.  In Section~\ref{subsec:nd} above we give an account of this proof from which we can extract the observation that the algorithm can be modified to work for an arbitrary input graph $G$ with the requirement that the algorithm \emph{may} simply reject the input if $G$ is not in $\CC$.  This suggests a tractable way of deciding $G\models \phi$ provided that $\phi$ defines a nowhere dense class.  Now the graphs is arbitrary, but the formula comes from a restricted class.  We formalise the result in the following theorem:

\begin{theorem} \label{T:fpfond_fpt}
Let $(Q, \kappa)$ be a problem that is slicewise first-order definable
and slicewise nowhere dense. Then $(Q, \kappa)$ is fixed-parameter tractable.
\end{theorem}
\begin{proof}

In the following, for ease of exposition, we assume that an instance of the problem consists of a graph $G$  and $\kappa(G) = i$ for some positive integer $i$.

\begin{description}
\item[Step 1: Compute $\phi$ and the parameters function.]
Since $(Q, \kappa)$ is slicewise first-order definable, we can compute from $i$ a first-order formula $\phi$ which defines the class of graphs $C_i = \{H \mid H \in Q \mbox{ and } \kappa(H) \leq i\}$.  Moreover, since $(Q, \kappa)$ is slicewise nowhere dense, we can compute from $i$ an algorithm that computes the parameter function $h$ for $C_i$.

\item[Step 2: Obtain $\hat \phi$ from $\phi$.]
By the Rank-Preserving Locality Theorem
(Theorem~\ref{T:rank_local}), we can compute from $\phi$ the formula $\hat \phi$ and a radius $r$.

\item[Step 3: Find a small cover $\cover$ for $G$.]
By Lemma~\ref{L:cover}, we can either find a cover
$\cover$ for $G$, or reject if the algorithm determine that $G \not\in C_i$.

\item[Step 4: Simulate Splitter game to compute
  $G \star^{q+1}_{\cover} q$.]
By  Lemma~\ref{L:structure} we obtain $G
\star^{q+1}_{\cover} q$ or reject if the algorithm determines that $G
\not\in C_i$.

\item[Step 5: Evaluate $\hat \phi$ on $G \star^{q+1}_{\cover} q$.]
Finally we evaluate $\hat \phi$ on $G \star_{\cover}^{q+1}
q$.
To do this, we need to solve the distance independent
set problem. We can do this by 
Lemma~\ref{L:indset}.

Since evaluating $\hat \phi$ on $G \star_{\cover}^{q+1} q$ is
equivalent to evaluating $\phi$ on $G$ this allows us to decide
whether $G \in Q$.
\end{description}
\end{proof}

\subsection{Applications}\label{subsec:applications}

In this Section we discuss some applications of Theorem~\ref{T:fpfond_fpt} that demonstrate its power.  We begin by considering simple edit distances.

\paragraph{Edit Distances}
A graph $G$ has \emph{deletion distance} $k$ to a class $\C$ if there exists a set $S$ of $k$ vertices in $G$ so that $G\setminus S\in \C$.   Suppose $(Q,\kappa)$ is  a parameterized graph problem.  We define the problem of deletion distance to $Q$ as follows:

\begin{mdframed}
 \textsc{Deletion Distance to $Q$}\\
 \textbf{Input: } A graph $G$ and $k, d \in \NAT$.\\
\textbf{Parameter:} $k+d$\\
 \textbf{Problem: } Does $G$ contain a set $S$ of  $k$ vertices so that $\kappa(G\setminus S) \leq d$ and $G \setminus S \in Q$?
\end{mdframed}

In many of the examples below, we define formulas of first-order logic by \emph{relativisation}.   For convenience, we define the notion here.

\begin{definition}
Let $\phi$ and $\psi(x)$ be first-order formulas, where $\psi$ has a distinguished free variable $x$ . The
relativisation of $\phi$ by $\psi(x)$, denoted $\phi^{[x.\psi]}$ is the formula obtained from $\phi$
by replacing all subformulas of the form
$\exists v \, \phi'$ in $\phi$ by $\exists v  (\psi[v/x] \land \phi' )$,
and all subformulas of the form $\forall v \, \phi'$ in $\phi$ by
$\forall v (\psi[v/x] \to \phi')$.  Here $\psi[v/x]$ denotes the result of replacing the free occurrences of $x$ in $\psi$ with $v$ in a suitable way avoiding capture.
\end{definition}
The key idea here is that if $\phi^{[x.\psi]}$ is true in a graph $G$
if $\phi$ is true in the subgraph of $G$ induced by the vertices that
satisfy $\psi(x)$.  Note that the variable $x$ that is free is $\psi$
is bound in $\phi^{[x.\psi]}$.  Other variables that appear free in
$\psi$ remain free in $\phi^{[x.\psi]}$.  We stress this as it is
needed in Proposition~\ref{P:tree-depth} where nested relativisations
are used.

\begin{proposition}\label{P:deletion}
If $(Q,\kappa)$ is slicewise nowhere dense and slicewise first-order definable then 
\textsc{Deletion Distance to $Q$} is $\FPT$.
\end{proposition}
\begin{proof}
  It suffices to show that \textsc{Deletion Distance to $Q$} is also slicewise nowhere dense and slicewise first-order definable.  For the latter, note that if $\phi_i$ is the first-order formula that defines the class of graphs $\C_i = \{G \mid \kappa(G) \leq i \mbox{ and } G \in Q\}$, then the class of graphs at deletion distance $k$ to $\C_i$ is given by:
\[
\exists w_1, \dots, w_{k} \phi_i^{[x.\theta_k]}
\]
where $\theta_k(x)$ is the formula $\bigwedge_{1\leq i\leq k} x \neq
w_i $.

To see that \textsc{Deletion Distance to $Q$} is also slicewise
nowhere dense, let $h$ be the parameter function for $Q$.  If the
graph $h(i,r)$ has $m$ vertices, then $K_m$ is not a
depth-$r$-minor of any graph in $\C_i$.  Then a graph which has
deletion distance $k$ to $\C_i$ cannot have $K_{m+k}$ as a depth-$r$-minor.
Indeed, suppose $K_{m+k} \preceq_r G$ and $G\setminus S \in \C_i$ where $S$ is a set of $k$ vertices.  Vertices from $S$
can appear in the images of at most $k$ vertices from  $K_{m+k}$ under
the minor map.  Thus, this minor map also witnesses that $K_m \preceq_r
G\setminus S$, a contradiction.
\end{proof}



Instead of deleting vertices, we can also consider editing the
graph by adding or deleting edges.    It is easily seen that we can modify a first-order formula $\phi$ to define the class of graphs $G$ that can be made to satisfy $\phi$ by $k$ edge additions or deletions.  
An added or deleted edge is
identified by a pair of vertices $uv \in E(G)$. If we want to add the
edge $uv$ we can replace all occurrences of  $E(w_1, w_2)$ in $\phi$ by:
\[
(w_1 = u \land w_2 = v) \lor (w_1 = v \land w_2 = u) \lor  E(w_1, w_2).
\]
Similarly we delete an edge $uv$ by replacing all occurrences of
$E(w_1, w_2)$ in $\phi$ by:
\[
(w_1 \neq u \lor w_2 \neq v) \land (w_1 \neq v \lor w_2 \neq u) \land E(w_1, w_2).
\]
Thus, an analogue of Proposition~\ref{P:deletion} is obtained for any combination of vertex and edge deletions and additions.
Golovach~\cite{Golovach:2014eq} proved that that editing a graph to
degree $d$ using at most $k$ edge additions/deletions is $\FPT$
parameterized by $k+d$.   Since the class of graphs of degree $d$ is first-order definable and nowhere dense for any $d$, the  result also follows directly from Theorem~\ref{T:fpfond_fpt}.

\paragraph{Tree-depth.}

Tree-depth is a graph parameter that lies 
between the widely studied parameters vertex cover number and tree
width.  It has interesting connections to nowhere dense graph
classes.  It is usually defined as follows:

\begin{definition}
The \emph{tree-depth} of a graph $G$, written $\td(G)$, is defined as
follows:
\[
\td(G) := \begin{cases}
 0,
 & \text{if }V(G) = \emptyset; \\
 1 + \min\{\td(G \setminus v) \mid v \in V(G)\},
 & \text{if $G$ is connected;} \\
 \max\{\td(H) \mid H \text{ a component of $G$}\},
 & \text{otherwise.}
\end{cases}
\]
\end{definition}

Note that a graph has tree-depth $k$
if and only if it has elimination distance $k$ to the class of empty graphs.  So one can
think of elimination distance as a natural generalisation of tree-depth.

It is known that the problem of determining the tree-depth of graph is $\FPT$, with tree-depth as the parameter (see~\cite[Theorem~7.2]{sparsity}).
We now give an alternative proof of this, using Theorem~\ref{T:fpfond_fpt}.  It is clear that for any $k$, the class of graphs of tree-depth at most $k$ is nowhere dense.
We show below that it is also first-order definable.

\begin{proposition}\label{P:tree-depth}
For each $k \in \NAT$ there is a first-order formula $\phi_k$ such
that a graph $G$ has tree-depth $k$ if and only if $G \models
\phi_k$.
\end{proposition}
\begin{proof}

We use the fact that in a graph of tree-depth less than $k$,
there are no paths of length greater than $2^k$.   This allows us, in
the inductive definition of tree-depth above, to replace the condition
of connectedness (which is not first-order definable) with a
first-order definable condition on vertices at distance at most $2^k$.

Let $\dist_d(u,v)$ denote the first-order formula with free variables
$u$ and $v$ that is satisfied by a pair of vertices in a graph
$G$ if, and only if, they have distance at most $d$ in $G$.  Note that
the formula  $\dist_d^{[x.x \neq w]}(u,v)$ is then a formula with three
free variables $u,v,w$ which defines those $u,v$ which have a path of
length $d$ in the graph obtained by deleting the vertex $w$.

We can now define the formula $\phi_k$ by induction. Only the empty graph
has tree-depth $0$, so $\phi_0 := \lnot\exists v(v = v)$.

Suppose that $\phi_k$ defines the graphs of tree-depth at most $k$,
let 
\[
\theta_{k} := (\forall u,v \, \dist_{2^{k+1}}(u,v)) \land \exists w (\phi_k^{[x.x \neq w]}).
\]
The formula $\theta_k$ defines the connected graphs of tree depth at
most $k+1$.  Indeed, the first conjunct ensures that the graph is
connected as no pair of vertices has distance greater than $2^{k+1}$
and the second conjunct ensures we can find a vertex $w$ whose removal yields a graph of
tree-depth at most $k$.

We can now define the formula $\phi_{k+1}$ as follows.
\[
\phi_{k+1} :=  (\forall u,v \, \dist_{2^{k+1}+1}(u,v) \to \dist_{2^{k+1}}(u,v)) \land
\forall w \theta_k^{[x.\dist_{2^{k+1}}(w,x)]}.
\]
The formula asserts that there are no pairs of vertices whose distance
is strictly greater than $2^{k+1}$ and that for every vertex $w$, the
formula $\theta_k$ holds in its connected component, namely those
vertices which are at distance at most  $2^{k+1}$ from $w$.
\end{proof}

While the proof of Proposition~\ref{P:deletion} shows that deletion distance to any slicewise first-order definable class is also slicewise first-order definable, Proposition~\ref{P:tree-depth} shows that elimination distance to the particular class of empty graphs is slicewise first-order definable.  It does not establish this more generally for elimination distance to any slicewise nowhere dense class.

\section{Elimination distance to classes characterised by excluded minors}
\label{S:minor}

In this section we show that determining the elimination distance of a
graph to a minor-closed class $\CC$ is $\FPT$ when 
parameterized by the elimination distance.  More generally, we
formulate the following parameterized problem where the forbidden
minors of $\CC$ are also part of the parameter.

\begin{mdframed}
 \textsc{Elimination Distance to Excluded Minors}\\
 \textbf{Input: } A graph $G$, a natural number $k \in \NAT$ and a set of graphs $M$\\
\textbf{Parameter:} $k+\sum_{G \in M}|G|$\\
 \textbf{Problem: } Does $G$ have elimination distance $k$ to the
 class $\C$ characterised by $M(\C) = M$?
\end{mdframed}

It is not difficult to show that the class of graphs which have
elimination distance $k$ to a minor-closed class $\C$ is also a
minor-closed class.  Indeed, this can be seen directly from an
alternative characterisation of elimination distance that we establish
below.  
The characterisation is in terms of the iterated closure of $\C$ under
the operation of disjoint unions and taking the class of apex graphs.
We introduce a piece of notation for this in the next definition.
Recall that we write $\apex{\C}$ for the class of all the apex graphs over $\C$, and that we write $\overline\C$ for the closure of $\C$ under disjoint unions.

\begin{definition}
For a  class of graphs $\C$, let $\C_0 := \C$, and $\C_{i+1} := \overline{\apex{\C_i}}$.
\end{definition}

We show next that the class $\C_k$ is exactly the class of graphs at
elimination distance $k$ from $\C$.

\begin{proposition} \label{P:char}
Let $\C$ be a class of graphs and $k \geq 0$. Then $\C_k$ is the class of all graphs with elimination distance at most $k$ to $\C$.
\end{proposition}
\begin{proof}
We prove this by induction. Only the graphs in $\C$ have elimination distance $0$ to $\C$, so the statement holds for $k = 0$.

Suppose the statement holds for $k$. If $G \in \C_{k+1}$, then $G$ is
a disjoint union of graphs $G_1, \dots, G_s$ from $\apex{\C_{k}}$, so
we can remove at most one vertex from each of the $G_i$ and obtain a
graph in $\C_{k}$.  Thus the elimination distance of $G$ to $\C_{k}$
is $1$, and by induction the elimination distance to $\C$ is
$k+1$. Conversely, if $G$ has elimination distance $k+1$ to $\C$, then
we can remove a vertex from each component of $G$ to obtain a graph
$G'$ with elimination distance $k$ to $\C$.  Using the induction hypothesis each component of $G'$ is in $\C_{k}$, and thus $G \in \C_{k+1}$.
\end{proof}

It is easy to see that if $\C$ is a minor-closed class of graphs
then so is $\C_k$ for any $k$.  Indeed, it is well-known that
$\apex{C}$ is minor-closed for any minor-closed $\C$, so we just need
to note that $\overline{\C}$ is also minor-closed.  But it is clear
that if $H$ is a minor of a graph $G$ that is the disjoint union of
graphs $G_1,\ldots,G_s$, then $H$ itself is the disjoint union of
(possibly empty) minors of  $G_1,\ldots,G_s$.  Thus, the class of graphs of elimination
distance at most $k$ to a minor-closed class $\C$ is itself
minor-closed.  We next show that we can construct the set of its
minimal excluded minors from the corresponding set for $\C$.

To obtain $M(\C_k)$, we need to iteratively compute $M(\apex{\C})$ and $M(\overline\C)$ from $M(\C)$.
Adler et al.~\cite{Adler:2008wm} show that from the set of minimal excluded minors $M(\C)$ of a class $\C$, we can compute
$M(\apex{\C})$:

\begin{theorem}[\cite{Adler:2008wm}, Theorem 5.1] \label{T:apex}
There is a computable function that takes the set of graphs
$M(\C)$ characterising a minor-closed class $\C$ to the set
$M(\apex{\C})$. 
\end{theorem}

 We next aim to  show that from $M(\C)$ we can also compute
compute $M(\overline{\C})$.  Together with Theorem~\ref{T:apex} this implies that from $M(\C)$ we
can compute $M(\C_k)$, the set of minimal excluded minors for the class of graphs with elimination
distance $k$ to $\C$.


We begin by characterising minor-closed classes that are closed under
disjoint unions in terms of the connectedness of their excluded minors.

\begin{lemma}
Let $\C$ be a class of graphs closed under taking minors. Then $\C$ is
closed under taking disjoint unions iff each graph in $M(\C)$ is
connected.
\end{lemma}
\begin{proof}
Let $\C$ be a minor-closed class of graphs, and let $M(\C)
= \{H_1, \dots, H_s\}$ be its set of minimal excluded minors.

Suppose each of the graphs in $M(\C)$ is connected. Let $H \in M(\C)$ and
let $G = G_1 \oplus \dots \oplus G_r$ be the disjoint union of graphs
$G_1, \dots, G_r \in \C$. Because $H$ is connected, we have that $H
\preceq G$ if and  only if $H \preceq G_i$ for one $1 \leq i \leq
r$. So, since all the $G_i \in \C$, we have $H \not\preceq G$ and thus
$G \in \C$. This shows that $\C$ is closed under taking disjoint
unions.

Conversely assume one of the graphs $H \in M(\C)$ is not
connected and let $A_1, \dots, A_t$ be its components . Then $A_1, \dots, A_t \in \C$, since each $A_i$ is a proper
minor of $H$, and $H$ is minor-minimal in the complement of $\C$. However, $A_1 \oplus \dots \oplus A_t = H \not\in
\C$.
\end{proof}

\begin{definition} \label{D:cc}
For a graph $G$ with connected components $G_1, \dots, G_r$, let $\H$
denote the set of \emph{connected} graphs $H$ with $V(H) = V(G)$ and
such that the subgraph of $H$ induced by $V(G_i)$ is exactly $G_i$.  
We define the \emph{connection closure} of $G$ to be 
the set of all minimal (under the subgraph relation) graphs in $\H$.
The connection closure of a set of graphs is the union of the connection closures of the graphs in the set.
\end{definition}
Note that if $G$ has $e$ edges and $m$ components, then any graph in
the connection closure of $G$ has exactly $e+m-1$ edges.  This is
because it has $G$ as a subgraph and in addition $m-1$ edges
corresponding to a tree on $m$ vertices connecting the $m$
components.

\begin{lemma} \label{L:disjoint_union}
Let $\C$ be a minor-closed class of graphs. Then $M(\overline\C)$ is
the set of minor-minimal graphs in the connection closure of $M(\C)$.
\end{lemma}
\begin{proof}
Let $\C$ be a minor-closed class of graphs, with $M(\C)$
its set of minimal excluded minors, and let $\hat{M}$ be the connection closure of $M(\C)$. 

Let $G$ be a graph such that $\hat H \not\preceq G$ for all $\hat H
\in \hat{M}$.  Suppose for a contradiction that $G$ is not a disjoint
union of graphs from $\C$.  
 Then there is a component $G'$ of $G$ that is not in $\C$ and
 therefore there is a graph $H \in M(\C)$ such that $H \preceq G'$.
We show that one of the graphs in the connection closure of $H$ is a minor of $G'$.


Let $\{w_1, \dots, w_s\}$ be the vertex set of $H$ and consider the image $T_1, \dots, T_s$
of the minor map from $H$ to $G'$. Let $T$ be a minimal subtree of $G'$ that contains
all of the $T_i$.  Such a tree must exist since $G'$ is connected.
Let $\hat H$ be the graph with the same vertex set as $H$, and an edge
between two vertices $w_i, w_j$ whenever either $w_iw_j \in E(H)$ or when there is a path between
$T_{w_i}$ and $T_{w_j}$ in $T$ that is disjoint from any $T_{w_k}$
with $w_i \neq w_k \neq w_j$.  We claim that $\hat H$ is in the
connection closure of $H$.  By construction, $\hat H$ is connected and
contains all components of $H$ as disjoint subgraphs, so we only need to
argue minimality.  $\hat H$ has no vertices besides those in $H$ so no
graph obtained by deleting a vertex would contain all components of
$H$ as subgraphs.  To see that no edge of $\hat H$ is superfluous, we
note it has exactly $e+m-1$ edges and thus no proper
subgraph could be connected and have all components of $H$ as disjoint
subgraphs. 
By the construction $\hat H \preceq G' \preceq G$, so by the transitivity of the minor relation we 
have that $\hat H \preceq G$.

Conversely let $G$ be an arbitrary graph and assume that $\hat H \in
\hat{M}$ and $\hat H \preceq G$. Because $\hat H$ is  connected, there
is a connected component $G'$ of $G$ such that $\hat H \preceq
G'$. 
Now there must be a graph $H \in M(\C)$ such that  $\hat H$ is in the
connection closure of $H$, and since $H$ is a subgraph of $\hat H$,  $H \preceq \hat H$.
Then, by the transitivity of the minor relation, $H \preceq G'$ and thus $G' \not\in \C$. Therefore $G$
is not a disjoint union of graphs from $\C$.
\end{proof}

Now we can put everything together and prove our main theorem: 

\begin{theorem}
There is a computable function which takes a set $M$ of excluded
minors characterising a minor-closed class $\C$ and $k \geq 0$ to the
set $M(\C_k)$. 
\end{theorem}
\begin{proof}
The proof is by induction. For $k = 0$, the set of minimal excluded minors of $\C_0$ is $M(\C_0) = M(\C)$, which is given.
For $k > 0$, we have that $\C_k = \overline{\apex{\C_{k-1}}}$. By the induction hypothesis we can compute $M(C_{k-1})$, by Theorem~\ref{T:apex} we can compute $M(\apex{\C_{k-1}})$ and using Lemma~\ref{L:disjoint_union} we can compute the connection closure of $M(\apex{\C_{k-1}})$ to obtain $M(\overline{\apex{\C_{k-1}}}) = M(\C_k)$.
\end{proof}

So by the Robertson-Seymour Theorem we have the following:
\begin{corollary}
Let $\C$ be a minor-closed graph class. Then the
problem \textsc{Elimination Distance to Excluded Minors} is $\FPT$.
\end{corollary}

\section{Conclusion}
\label{S:conclusion}

We are motivated by the study of the fixed-parameter tractability of edit distances in graphs.  Specifically, we are interested in edit distances such as the number of vertex or edge deletions, as well as more involved measures like elimination distance.  Aiming at studying general techniques for establishing tractability, we establish an algorithmic meta-theorem showing that any slicewise first-order definable and slicewise nowhere dense problem is $\FPT$.  This yields, for instance, the tractability of counting the number of vertex and edge deletions to a class of bounded degree.  As a second result, we establish that determining elimination distance to any minor-closed class is $\FPT$, answering an open question of~\cite{BD14}.

A natural open question raised by these two results is whether elimination distance to the class of graphs of degree $d$ is $\FPT$.  When $d$ is $0$, this is just the tree-depth of a graph, and this case is covered by our first result.  For positive values of $d$, it is not clear whether elimination distance is first-order definable.  Indeed, a more general version of the question is whether for any nowhere dense and first-order definable $\CC$, elimination distance to $\CC$ is $\FPT$.

Another interesting case that seems closely related to our methods, but is not an immediate consequence is that of classes that are given by first-order interpretations from nowhere dense classes of graphs.  For instance, consider the problem of determining the deletion distance of a graph to a disjoint union of complete graphs.  This problem, known as the cluster vertex deletion problem is known to be $\FPT$ (see~\cite{HKMN10}).
The class of graphs that are disjoint unions of cliques is first-order definable but certainly not nowhere dense and so the method of Section~\ref{S:nd} does not directly apply.  However, this class is easily shown to be interpretable in the nowhere dense class of forests of height~$1$.  Can this fact be used to adapt the methods of Section~\ref{S:nd} to this class?

\newpage
\bibliographystyle{amsplain}
\bibliography{distances_fpt.bib}

\end{document}